\documentclass[a4paper,11pt]{article}
\usepackage[english]{babel}
\usepackage{amsthm}
\usepackage[dvips]{graphicx}
\usepackage{amssymb,amsmath,eepic,epic,colordvi,mathrsfs}
\usepackage{float}
\usepackage{fullpage}
\usepackage{yfonts}

\newcommand{\ignore}[1]{}

\DeclareMathAlphabet{\mathpzc}{OT1}{pzc}{m}{it}

\theoremstyle {plain}

\newtheorem{theorem}{Theorem}[section]
\newtheorem{proposition}[theorem]{Proposition}
\newtheorem {corollary}[theorem]{Corollary}
\newtheorem{definition}[theorem]{Definition}
\newtheorem{lemma}[theorem]{Lemma}

\title{Simple Executions of Snapshot Implementations}
\author{Gal Amram, Lior Mizrahi and Gera Weiss\\ Department of Computer Science,\\Ben-Gurion University, Beer Sheva Israel, 84105.\\ \{galamra,liormizr,geraw\}@cs.bgu.ac.il }

\begin{document}
\newcommand{\bigq}{\mathbb{Q}}
\newcommand{\bb}{\bfseries\large}
\newcommand{\class}{\mathcal}
\newcommand{\mbi}[1]{\textbf{\em{#1}}}
\newcommand{\bi}[1]{\textbf{\textit{#1}}}
\newcommand{\ita}[1]{\textit{#1}}
\newcommand{\p}{\cdot}
\newcommand{\req}{\emptyset}
\newcommand{\init}{\mathcal{J}}
\newcommand{\power}[1]{\class{P}(#1)}
\newcommand{\bigpower}[2]{\class{P}^{#1}(#2)}
\newcommand{\symetric}{\bigtriangleup}
\newcommand{\gx}[3]{\Gamma_{(#1,#2)}(#3)}
\newcommand{\bd}[3]{\Delta_{(#1,#2)}(#3)}
\newcommand{\bp}[3]{\Phi_{(#1,#2)}(#3)}
\newcommand{\eab}{\eta_{\alpha \beta}}
\newcommand{\sab}{\sigma_{\alpha\beta}}
\newcommand{\dab}{\class{D}_{\alpha\beta}}
\newcommand{\cab}{\class{C}_{\alpha\beta}}
\newcommand{\C}[1]{\class{#1}}
\newcommand{\ol}{\overline}
\newcommand{\into}{\longrightarrow}
\newcommand{\ifff}{\Longleftrightarrow}
\newcommand{\imply}{\Longrightarrow}
\newcommand{\since}{\Longleftarrow}
\newcommand{\la}{\langle}
\newcommand{\ra}{\rangle}
\newcommand{\andd}{\wedge}
\newcommand{\orr}{\vee}
\newcommand{\ft}{(\C F \C T)_Q}
\newcommand{\pathopen}[3]{#1_{(\ol{#2},#3)}}
\newcommand{\fpi}[3]{#1_*(\Pi(#2,#3))}
\newcommand{\xpath}[2]{(\ol {#1},#2)}
\newcommand{\qn}[2]{\Lambda(#1,#2)}
\newcommand{\jh}{\class J \class H}
\newcommand{\bos}{\boldsymbol}
\newcommand{\f}{\mathbb}
\newcommand{\rss}{\la\psi_i:i\in \lambda\ra}
\newcommand{\ito}[1]{\overset{#1}{\longrightarrow}}
\newcommand{\pcomplex}[2]{\class #1(\class #2)}
\newcommand{\set}{{\sf set }}
\newcommand{\setv}[1]{{\sf set($#1$)}}
\newcommand{\compute}{{\sf compute }}
\newcommand{\Keywords}[1]{\par\noindent 
{\small{\bfseries Keywords\/}: #1}}
\newcommand{\updatev}{{\sf update($v$)}}
\newcommand{\update}{{\sf update }}
\newcommand{\scan}{{\sf scan }}
\newcommand{\vectoralpha}[1]{\alpha_0^#1,\dots,\alpha_{n-1}^#1}
\newcommand{\valsvectoralpha}[2]{(val(\alpha_0^#1(#2)),\dots,val(\alpha_{n-1}^#1(#2)))}
\newcommand{\emptyvectoralpha}{\alpha_0,\dots,\alpha_{n-1}}
\newcommand{\alphai}[1]{\alpha_#1:\text{\scan events}\into p_#1\text{-\update events}}
\newcommand{\st}{sched(\tau)}
\newcommand{\updatex}[1]{{\sf update($#1$)}}
\newcommand{\snap}{ \scriptstyle\mathcal{SNAP} }

\renewcommand\refname{6 \ \  References}

\maketitle
\begin{abstract}
The well known snapshot primitive in concurrent programming allows for $n$-asynchronous processes to write values to an array of single-writer registers and, for each process, to take a snapshot of these registers. In this paper we provide a formulation of the well known linearizability condition for snapshot algorithms in terms of the existence of certain mathematical functions. In addition, we identify a simplifying property of snapshot implementations we call ``schedule-based algorithms". This property is natural to assume in the sense that as far as we know, every published snapshot algorithm is schedule-based. Based on this, we prove that when dealing with schedule-based algorithms, it suffices to consider only a small class of very simple executions to prove or disprove correctness in terms of linearizability. We believe that the ideas developed in this paper may help to design automatic verification of snapshot algorithms. Since verifying linearizability was recently proved to be EXPSPACE-complete, focusing on unique objects (snapshot in our case) can potentially lead to designing restricted, but feasible verification methods.
\end{abstract}

\section{Introduction}
 
The snapshot object was introduced by Afek et al. \cite{snap1,snap2}, independently by Anderson \cite{comp} and by Aspens and Herlihy \cite{PRAM}. The snapshot object is shared by $n$ processes, $p_0,\dots,p_{n-1}$. This object is divided into $n$ segments when the $i$-th segment is ``owned" by process $p_i$. Each process $p_i$ can write values to its segment by invoking an \updatev \ operation with an argument $v$ taken from some fixed set of values $Vals$. In addition, each process can scan the entire array by invoking a \scan operation. Thus, \scan returns a vector consisting of $n$ elements from $Vals$. The snapshot object is an efficient tool for achieving synchronization between $n$ processes in the shared memory model (see chapter 9 in \cite{DC} for exact definitions), since it allows the processes to scan the entire shared memory\footnote{In this model it is suffice to assume that each process use only one single-writer register.} at an atomic action. Therefore, it is not surprising that the snapshot object is so well-studied, especially due to the fact that it can be implemented using only single-writer registers.
 
In \cite{snap1},\cite{comp} and \cite{PRAM}, while introducing the snapshot object, the correctness criterion adopted by the authors is the Linearizability criterion \cite{Lin} which is, nowadays the standard correctness condition for implementation of concurrent objects. Informally, Linearizability is the requirement that in any execution, each procedure execution can be identified with a unique moment during its actual execution, such that this identification yields a correct sequential execution (according to the specification of the object). The importance of this criterion is that it ensures an execution appears to a user as if it is sequential. This stands in contrast to other correctness conditions. For example, this property does not hold if only sequential consistency \cite{seq} is required. As linearizability and sequential consistency are the main correctness criteria accepted by researchers (see \cite{seq-lin} for detailed discussion), it is natural that the Linearizability criterion is widely adopted, as many authors claim \cite{DGH},\cite{GHW},\cite{GY}, \cite{Roman}, \cite{Vaf}.
 
The Linearizability criterion successfully formulates what one would consider as ``good behavior" of a concurrent system. Due to the complex nature of distributed systems, the research in the field of linearizability is deep and complicated. We see three aspects concerning this issue
 
\begin{enumerate}
\item Implementing concurrent objects is hard. One can use the trivial solution and lock the system before every operation. However, it seems that avoiding such trivial solutions is solely at the hand of experts and researchers.
\item Proving correctness of linearizable implementations is difficult. Examining known-results in literature reveals that in many occasions, proofs tend to be long and technical. Moreover, many times proofs include clever and sophisticated ideas, so finding a correct implementation is sometimes only half of the work required of the programmer.

\item Automatic verification of linearizability is a hard problem. In general, it is undecidable \cite{undecidable}, and if the number of processes is fixed and all methods are finite, the problem is EXSPSPACE-complete \cite{expspace}.
\end{enumerate}

This paper includes three contributions. In theorem
\ref{main-prop-of-section-1-correctness-condition}, we provide a necessary and sufficient condition
for linearizability of executions of snapshot implementations. In definition
\ref{algorithm-assumption} we introduce the notion of schedule-based snapshot algorithm. This notion
captures a natural property of concurrent implementations and in fact, we are not familiar with any
published snapshot implementation which is not scheduled-based. Finally, in what we consider as our
main contribution, we prove in theorem \ref{simple-executions-are-suffice} that a schedule-based
snapshot algorithm is correct if all its simple executions are correct. A simple execution is an
execution in which all processes, excluding two processes, invoke only \updatex 0 and \scan
operations. The remaining two processes may also execute an \updatex 1 procedures, but once a
process executes an \updatex 1 procedure, it is not allowed to invoke an \updatex 0 procedure again
for the rest of the execution.

Informally, a snapshot algorithm is scheduled-based if at any execution, the values that a \scan operation returns depend on the interleaving of the actions performed by the processes, and not on the actual values that the \update procedures wrote to the segments of the snapshot object. To illustrate the idea behind this notion, consider an execution in which, whenever a process executes an \update procedure, it invokes \updatex m when $m$ is a counter that counts the number of \update 
operations executed by the process. Now assume a different execution in which the processes take steps at the same order, but instead of calling \updatex m, the $m$-th \update operation of each process is \updatex {m+1}. In this case, we expect that if a \scan operation at the first execution returns $(k_1,\dots,k_n)$, then there is a \scan 
operation at the second execution, that occurred at the ``same time" and returned $(k_1+1,\dots,k_n+1)$. This is a natural property to assume, since the snapshot object deals with synchronization between reads and writes, and the actual values that the processes write to the segments are immaterial. It can be observed that authors refer to their algorithms as schedule-based without formulating exactly the scheduled-based notion. When Attiya, Herlihy and Rachman \cite{LA} write: 
\begin{quote}
we can ignore the real values written to the segments and refer only to the sequence numbers\footnote{These sequence numbers counts the number of \update operations.} that are written there.
\end{quote}
they mean that their algorithm is scheduled based. We understand their statement in the following manner: since the values returned by \scan operations depend on the interleaving of the execution, but not on the actual values written to the segments, it suffice to assume that each process counts the number of \update operations, and write the value of this counter into its segment.

 We do not claim that any snapshot algorithm is schedule-based and in fact, it is not difficult to transform a correct schedule-based implementation into a correct not-schedule-based algorithm. But since (for the best of our knowledge) every published snapshot implementations is scheduled-based, in practice, a non-schedule-based algorithm is likely to be an algorithm obtained by optimization of some schedule-based implementation.

We mentioned three difficulties concerning Linearizability: constructing correct implementations is hard, proving correctness is difficult, and the problem is EXPSPACE-complete. We demonstrate now how our contributions address these three issues.

\begin{enumerate}

\item {\bfseries A necessary and sufficient condition for correctness of executions of snapshot implementation.} Our condition provides an alternative framework for designing correct snapshot implementation and for proving correctness of snapshot implementations. Instead of trying to achieve linearizability, one needs to try to satisfy our condition. Possibly, some programmers will find our condition easier to work with.

During the writing process of this paper, we were surprised to find out that our condition is similar to the condition in Anderson's ``shrinking lemma" \cite{comp}. However, we believe that our condition is more natural and it provides a better framework for programmers than the shrinking lemma. Our condition deals with the existence of functions between \scan events and \update events that satisfy some properties. Informally, for each $i<n$, we have a function $\alpha_i$ such that if $S$ is a \scan event, $\alpha_i(S)$ is the $p_i$-\update event in which $p_i$ wrote to its segment the value read by $S$. It is clear that these functions must satisfy some properties. For example, there cannot be a $p_i$-\update event between $\alpha_i(S)$ and $S$. In a similar way, our properties classify all the ``bugs" that might occur in an execution. Thus, while proving correctness, it is reasonable that the programmer will be able define the functions $\alpha_0,\dots,\alpha_{n-1}$. She just need to explain which \update events wrote the values returned by a \scan event. To summary, we provide the programmer with a list of properties, and She needs to check that these bugs never arise in any execution, while writing the code or while proving correctness.

\item {\bfseries Scheduled-based algorithms.} Recently, verifying Linearizability was proved to be EXPSPACE-complete \cite{expspace}. Thus, complete verification is infeasible. One way to overcome this gap is to check for errors in short executions  \cite{BDMT},\cite{LCLS},\cite{VY}. Another way is to ask the user to specify the linearization points \cite{BLMRM},\cite{Vaflp}. A remarkable result can be found in \cite{Vaf}. The key idea in \cite{Vaf} is to assume that the linearization points of the algorithm satisfy some properties. Since the general case is EXPSPACE-complete, it is necessary to adapt such assumptions, although the assumption in \cite{Vaf} excludes some known implementations, such as the queue implementation in \cite{Lin}. Here we suggest the schedule-based property. We see potential in this natural assumption, and it could lead to results concerning automatic verification of algorithms with reasonable time-complexity.

We also suggest to look at specific objects. The general case might be difficult, but it is possible that for some specific objects, verification can be feasible. In this paper we focus on the snapshot object, but it is straightforward to generalize our notion for other objects as well, as long as the values returned by operations depend on the ordering of method invocations and not on the exact arguments (for example: stack, queue, etc. in contrast to test-and-set). Thus, the ideas we develop in this paper may lead to similar results regarding other objects and data-structures, and may lead to improved verification techniques.

\item {\bfseries Reduction to simple executions.} Alur et al. \cite{AMP} showed that linearizability is decidable when the number of processes is fixed and the implementation is finite (no unbounded registers are used such as integers, etc.). Regarding the snapshot object, it is possible that an implementation is infinite only because $Vals$ is an infinite set. The traditional way to overcome this issue, is to check correctness under the assumption that the processes invoke only $\scan$, \updatex 0 or \updatex 1 operations. In theorem \ref{reduction-to-simple-executions} we prove that this assumption is suffice for schedule-based implementations. Therefore, we conclude that if $\snap$ is a schedule-based snapshot algorithm, and if only finitely many configurations of $\snap$ are reachable if the processes execute operations from $\{\scan$, \updatex 0, $\sf{update(1)}\}$, then it is decidable to determine if $\snap$ is correct. This hold although the verification approach in \cite{AMP} cannot be applied on $\snap$ directly. Moreover, our reduction to simple executions also reduces the running time of the verification procedure in \cite{AMP} (in compare to the traditional approach mentioned above). Furthermore, our reduction shows that under some natural assumptions it suffice to consider only a small and simple class of executions. We believe that there is high potential in this reduction for obtaining a polynomial verification method of  schedule-based snapshot algorithms. 

In addition, when one tries to develop a snapshot implementation, naturally, his construction is likely to result in a schedule-based implementation. Thus, since we prove that is suffice to look at simple executions, we provide another framework for programmers. Instead of concerning that every execution is linearizable, one needs to consider only simple executions of the implementations. Therefore, our result can help designing correct implementations and can ease the process of writing proofs.

\end{enumerate}

\section{Preliminaries}

\subsection{Executions of Snapshot Algorithms}
A snapshot algorithm $\snap$ is an implementation of two methods: \update and {$\scan$}. \update gets as argument a value from a known fixed set of values, $Vals$, and \scan returns a vector of $n$ values from the set $Vals$, where $n$ is the number of processes. Formally, each method is modeled as a transition system, and a process is a transition system that nondeterministically executes \scan and $\update(v)$ operations with argument $v\in v$. More precisely, from the initial state of process $p_i$ (which is a transition system) there are arrows for each operation \scan or $\update(v)$, and each last action in a method ends in  the initial state of $p_i$. For a fixed number of processes $n$, we identify an algorithm $\snap$ with the parallel composition of the processes i.e. ${\snap}=p_0||p_1||\dots||p_{n-1}$ (see chapter 2 in \cite{MCbook}).

An execution $\tau$ of $\snap$ is a finite sequence of actions (named execution fragment in \cite{MCbook}) that the processes execute according to the code of the algorithm $\snap$. In an execution $\tau$, some of the methods invocations return and some are not. We say that an operation is {\textbf{\emph{complete}}}, if the process that executed the operation has executed all the commands and returned. Otherwise, the operation is said to be \bi{pending}. For a process $p_i$, each action by $p_i$ is also named an action or a low level event, and each operation also named a high level event (see \cite{Lamport} for further discussion). When the context is clear, we use the term \ita{event} without specifying if it is a low level or a high level event. Formally, a complete $p_i$-event in an execution $\tau$ is a pair $(s,t)\in \f N\times \f N$  such that
\begin{enumerate}
\item $s<t$.
\item $\tau(s)$ is the first action by $p_i$ of an operation.
\item $\tau(t)$ is the last action by $p_i$ of an operation. 
\item For each $s<l<t$, $\tau(l)$ is not a first action of an operation by $p_i$.
\end{enumerate}
Since pending operations have no last action, we define a pending event to be a pair $(s,\infty)$, $s\in \f N$ so that:
\begin{enumerate}
\item $\tau(s)$ is the first action by $p_i$ of an operation.
\item For each $s<l$, if $\tau(l)$ is defined, then it is not a first action of an operation by $p_i$.
\end{enumerate}
A high level event is either a \scan event or an \update event. For a high level event $E$, we also write that $E$ is a $p_i$-\scan event or a $p_i$-\update event for denoting which process executed the operation $E$.

For an execution $\tau$, $complete(\tau)$ denotes the set of all complete high level events in $\tau$, and $events(\tau)$ is the set of all high level events in $\tau$, pending and complete. Clearly, $complete(\tau)\subseteq events(\tau)$. In addition, if $E$ is an \update event we use $val_\tau(E)$ to denote the argument with which $E$ has been invoked, and if $E$ is a complete \scan event, $val_\tau(E)$ is the $n$-elements vector that $E$ returns. In addition, if $E$ is a complete \scan event we use $val_{\tau:i}(E)$ to denote the element at the $i$-th entry of $val_\tau(E)$. In case that $\tau$ is clear from the context, we  use $val(E)$ and $val_i(E)$ instead of $val_\tau(E)$ and $val_{\tau:i}(E)$. 



The low level events in an execution $\tau$ are linearly ordered by the precedence relation, $<$. We naturally extend this relation to high level events. For two high level events $E_1=(s_1,t_1)$ and $E_2=(s_2,t_2)$ we write $E_1<E_2$ if $t_1<s_2$, and we say in this case that $E_1$ precedes $E_2$ and that $E_2$ follows $E_1$. Note that no high level event follows a pending operation. Although $<$ is a linear ordering over the set of low level events, in many cases, $<$ is only a partial ordering over the set of high level events since it is possible that for two high level events $E_1$ and $E_2$, neither $E_1<E_2$ or $E_2<E_1$. Such high level events are said to be concurrent. 
$<$ also  relates low level events with high level events as follows: if $E=(s,t)$ is an high level event and $e=\tau(l)$ a low level event, we write $e<E$ if $l<s$ and $E<e$ is $t<l$. Furthermore, if $s\leq l\leq t$ and both $E$ and $e$ are $p_i$-events for a process $p_i$, then we write $e\in E$.   

For the purpose of our discussion, for simplicity, we assume that in any execution $\tau$ each process executes an initial \update operation in which the process writes the initial values to the registers (or just perform an initialization, when the exact form of the initialization depends on the communication media). Thus, we assume that in each execution $\tau$ there are $n$ initial \update events that precede any other high level event. These \update events are not necessarily follow the code of the algorithm, but they are considered as high level events in any execution.

\subsection{Linearizability}

Linearizability is the standard correctness condition for implementations of concurrent objects \cite{Lin}. Roughly speaking, an execution is linearizable if each operation can be seen as if it was executed in a unique instantaneous moment (the linearization point of the operation), during its actual execution. The requirement is that the identification of the high level events with their linearization points, yields a sequential execution that satisfies the correctness condition of the object: the sequential specification.

In an execution $\tau$, some operations are complete and some are pending. Some of the pending operations has affected the system and some may be neglected. Thus, the linearizability condition described above relates to all the complete operations in addition to some of the pending operations.

Now we describe the requirement  formally. Let $\tau$ be an execution of a snapshot algorithm $\snap$, and let $<$ be the precedence relation defined over $events(\tau)$. $\tau$ is linearizable if there is a set of events
$$complete(\tau)\subseteq \mathcal{E}  \subseteq events(\tau)$$ 
and a linear ordering $\prec$ on $\class E$ that extends $<$, so that the linear ordering $(\class E,\prec)$ satisfies the sequential specification of the snapshot object, presented below in figure \ref{snapshot-specification}. 


\begin{figure}[H]

\fbox{
\begin{minipage}[t]{140mm}

\begin{itemize}
\item[1.] The procedure executions are partitioned into {\sf update($v$)} and {\sf scan} operations, and are totally ordered by $\prec$.
$n$ initial {\updatev} operations are assumed, each initial {\updatev} operation has been executed by a different process. These operations precede all other operations in $\prec$.

\item[2.] For a \scan event $S$, let $U_i$ denote the maximal \update operation executed by $p_i$ such that $U_i\prec S$ thus $val(S)=(val(U_0),\dots,val(U_{n-1}))$.  

\end{itemize}

\end{minipage}
}
\caption{The snapshot sequential specification.}
\label{snapshot-specification} 
\end{figure} 

Therefore, an execution $\tau$ of a snapshot algorithm $\snap$ is said to be correct if it is linearizable, and a snapshot algorithm $\snap$ is correct if all of its executions are correct.

\subsection{Schedule-Based Algorithms}

\ignore{Verifying linearizable implementations is known to be $EXPSAPCE$- complete \cite{}. 
We present a polynomial verification of finite snapshot implementations, under some natural assumptions on the algorithm that is verified to be correct. Informally, the property that the algorithm is assumed to satisfy, is that the values that the (complete) \scan operations return are a matter of scheduling and they do not depend on the actual values with which the \update operations have been invoked. For example, we consider an execution of a snapshot algorithm, illustrated in Figure \ref{first-execution}.}

In section 4 we show that for a snapshot algorithm $\snap$, if $\snap$ is scheduled-based, then $\snap$ is correct iff all its simple executions are correct. Roughly speaking, an algorithm is scheduled-based if at any of its executions, the values that the (complete) \scan operations return are a matter of scheduling and they do not depend on the actual values with which the \update operations have been invoked. As an example, we consider an execution of a snapshot algorithm, illustrated in Figure \ref{first-execution}.

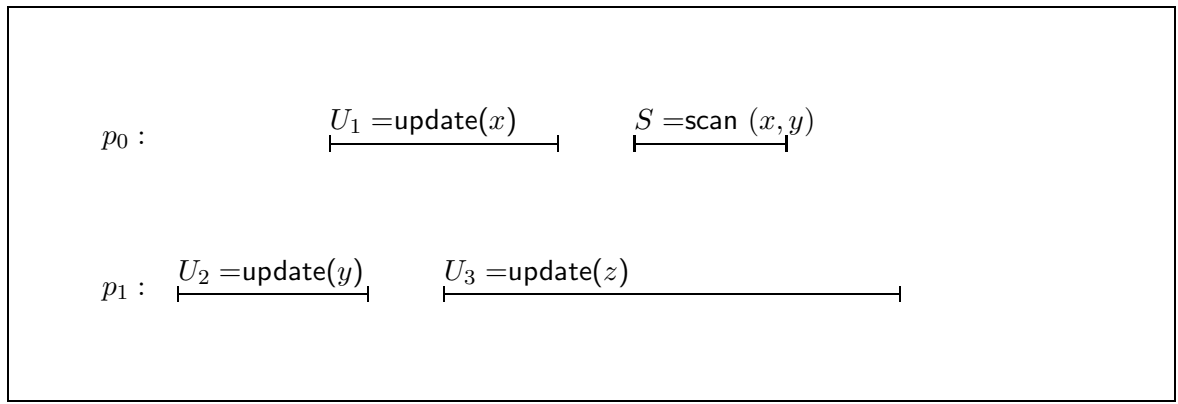
\begin{figure}[H]
\fbox{
\begin{minipage}[t]{15cm}

\setlength{\unitlength}{1cm}
\begin{picture}(5,5)

\put(1,3.3){$p_0:$}

\put(4,3.3){\line(1,0){3}}
\put(4,3.5){$U_1=$\updatex x}

\put(4,3.2){\line(0,1){0.2}}
\put(7,3.2){\line(0,1){0.2}}

\put(8,3.3){\line(1,0){2}}
\put(8,3.5){$S=$\scan$(x,y)$}

\put(8,3.2){\line(0,1){0.2}}
\put(10,3.2){\line(0,1){0.2}}

\put(1,1.3){$p_1:$}

\put(2,1.3){\line(1,0){2.5}}
\put(2,1.5){$U_2=$\updatex y}

\put(2,1.2){\line(0,1){0.2}}
\put(4.5,1.2){\line(0,1){0.2}}

\put(5.5,1.3){\line(1,0){6}}
\put(5.5,1.5){$U_3=$\updatex z}

\put(5.5,1.2){\line(0,1){0.2}}
\put(11.5,1.2){\line(0,1){0.2}}

\end{picture}
\end{minipage}}
\caption{first execution}
\label{first-execution}
\end{figure}  

In this execution $p_0$ executes an {\updatex x} operation $U_1$, and then executes a \scan operation $S$, which returns $(x,y)$. In addition, $p_1$ executes an {\updatex y} operation, $U_2$, and then an {\updatex z} operation, $U_3$. As the \scan operation, $S$, returns $(x,y)$ we have
\begin{equation}
val(S)=(val(U_1),val(U_2)).
\label{equation-for-similar-executions}
\end{equation}

The schedule-based property assumes that equation \ref{equation-for-similar-executions} holds due to the schedule of the execution and the operations that the process execute, but not on the values that the \update operations are invoked with (namely, $x,y$ and $z$). For example, if we let the processes operate in the same order as in the execution presented in Figure \ref{first-execution} and to execute the same operations only with different arguments, we shall get a similar execution as presented in the Figure \ref{second-execution}.

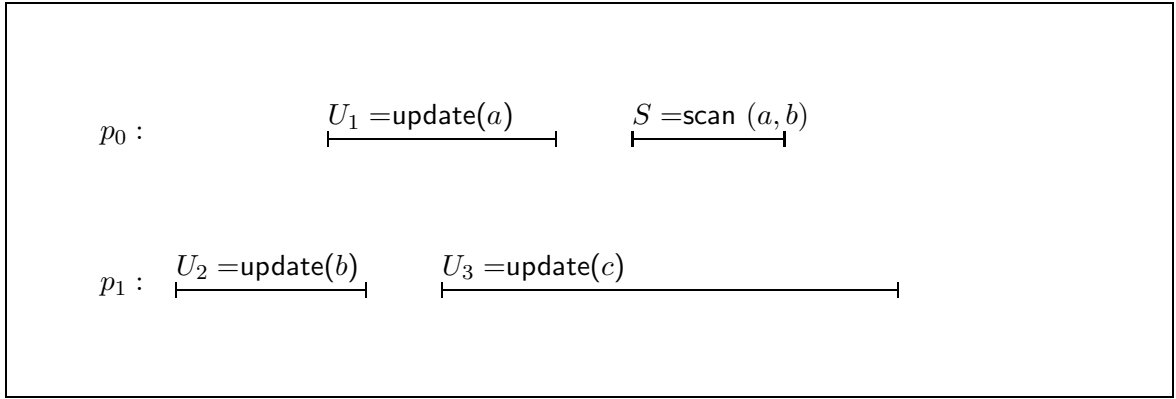
\begin{figure}[H]
\fbox{
\begin{minipage}[t]{15cm}

\setlength{\unitlength}{1cm}
\begin{picture}(5,5)

\put(1,3.3){$p_0:$}

\put(4,3.3){\line(1,0){3}}
\put(4,3.5){$U_1=$\updatex a}

\put(4,3.2){\line(0,1){0.2}}
\put(7,3.2){\line(0,1){0.2}}

\put(8,3.3){\line(1,0){2}}
\put(8,3.5){$S=$\scan$(a,b)$}

\put(8,3.2){\line(0,1){0.2}}
\put(10,3.2){\line(0,1){0.2}}

\put(1,1.3){$p_1:$}

\put(2,1.3){\line(1,0){2.5}}
\put(2,1.5){$U_2=$\updatex b}

\put(2,1.2){\line(0,1){0.2}}
\put(4.5,1.2){\line(0,1){0.2}}

\put(5.5,1.3){\line(1,0){6}}
\put(5.5,1.5){$U_3=$\updatex c}

\put(5.5,1.2){\line(0,1){0.2}}
\put(11.5,1.2){\line(0,1){0.2}}

\end{picture}
\end{minipage}}
\caption{second execution}
\label{second-execution}
\end{figure}  

As the executions in figures \ref{first-execution} and \ref{second-execution} are similar, we expect that equation \ref{equation-for-similar-executions} will hold in both, or in none of this two executions. Of course, this is just a unique example and formally, we require that for any two similar executions and for any \scan event, any equation that resemble equation \ref{equation-for-similar-executions} will hold in both of the executions or in none. For providing the exact definition, we first formulate what we precisely mean when by saying that two executions are similar.

\begin{definition}
\label{similar-executions}
Two execution $\tau$ and $\tau'$ are similar if
\begin{enumerate}
\item Both are of the same length.

\item For each $l$, $\tau(l)$ and $\tau'(l)$ are actions by the same process.

\item For each process $p_i$, the $l$-th $p_i$-operation in $\tau$ is a \scan event iff the $l$-th $p_i$-operation in $\tau'$ is a \scan event.


\end{enumerate}
\end{definition}

Thus, in similar executions the processes operate at the same order and they execute the same procedures. Similar executions only differ by the values that the \update operations are invoked with (and by the values that \scan operations return, due to the difference in the values with which \update were invoked). The property we want to formulate is that in similar executions, the operations that correspond to each other start and end at the same time, and that the \scan operations return the value wrote by corresponding \update events. As an example, in Figure \ref{first-execution} the first $p_1$-\scan operation returns the values of the first $p_0$-\update operation and the first $p_1$-\update operation. As the execution in Figure \ref{second-execution} is similar to this execution, the operations invoked and return at the ``same time", and the \scan event also returns the values of the first \update operations.


\begin{definition}\label{algorithm-assumption}
A snapshot algorithm $\snap$ is said to be a \bi{schedule-based} algorithm (or just sb-algorithm) if for any execution $\tau$, there are functions $\alpha_0^\tau,\dots,\alpha_{n-1}^\tau$
$$\alpha_i^\tau :\text{ complete } \scan\text{ events}\into p_i\text{-}\update\text{ events}$$
such that for any execution $\tau'$ similar to $\tau$:
\begin{enumerate}
\item $(s,t)\in \f N\times (\f N\cup\{\infty\})$ is a $p_i$-\scan ($\update$) event in $\tau$ iff it is a $p_i$-\scan ($\update$) event in $\tau'$.
\item For complete \scan event $E$
$$val_{\tau'}(E)=(val_{\tau'}(\alpha_0^\tau(E)),\dots,val_{\tau'}(\alpha_{n-1}^\tau(E))).$$ 
\end{enumerate}
\end{definition}

Note that if $\tau$ and $\tau'$ are similar and $E=(s,t)$, then by the first requirement $E$ is a scan event in $\tau$ iff it is a \scan event in $\tau'$. Moreover, for each $i<n$ $\alpha^\tau _i(E)$ is a $p_i$-\update event in both the executions $\tau$ and $\tau'$ thus $val_{\tau'}(\alpha^\tau_i(E))$ is well defined. 

Of course, not every snapshot algorithm is an sb-algorithm and in fact, it is possible to transform a correct snapshot algorithm into a correct algorithm which is not schedule based. However, as the snapshot problem deals with synchronization between processes, by the essence of the problem the schedule based property is very natural to assume. Indeed, we are not familiar with any published snapshot algorithm which is not schedule based and by the reasons described here, it seems unnatural to come up with such an algorithm.

\subsection{Finite Implementations}

Alur, et al. proved in \cite{AMP} that it is decidable to determine if an implementation is linearizable with respect to a sequential specification, in case that number of processes is fixed and that all methods can be modeled as finite transition systems. As a consequence, since we do not make assumptions on the size of the set $Vals$ (it can be infinite), Alur et al.  approach \cite{} cannot be applied on snapshot implementations. (unless we assume that $Vals$ is finite and then it can be applied on finite implementations.)

We observe that some snapshot implementations are infinite {\bfseries only because} $Vals$ is infinite set. As an example, consider the classical snapshot algorithms in \cite{snap2}. The algorithm in section 3 is clearly infinite since each process use a field named $seq$ which counts the number of \update events. Now, the algorithm in section 4 is also infinite since each register $r_i$ store a value $data\in Vals$ and possibly $|Vals|=\infty$. But, in the second case, if all \update events are invoked with values taken from some finite range, the registers may store only finitely many different values and we get a finite algorithm. 

For the purpose of our discussion, we say that a snapshot implementation is finite, if it is finite in case that the processes execute only $\scan$,\updatex 0 and \updatex 1 operations. Coming back to our example, the algorithm in section 3 in \cite{snap2} is infinite while the one in section 4 is finite. According to theorem \ref{reduction-to-simple-executions}, it suffice to focus on simple executions, regarding sb-algorithms. We conclude that linearizability of finite snapshot sb-algorithms is decidable.

\section{A necessary and sufficient condition for the correctness of a snapshot algorithm}

In this section we present a necessary and sufficient condition for the correctness of an execution $\tau$ of a snapshot algorithm $\snap$. Here $\snap$ is not assumed to be an sb-algorithm. The condition we describe is equivalent to linearizability of any execution of any snapshot implementation $\snap$. Our condition relies upon the existence of $n$ function $\emptyvectoralpha$, 
$$\alphai i$$
that satisfy several properties. In the next definition we define the properties that the functions are required to satisfy.

\begin{definition}
\label{definition-of-correct-functions}
Let $\{\alpha_i : i\in n\}$ be a set of
$n$ functions such that
$$\alphai i.$$
We say that the functions $\alpha_0,\dots,\alpha_{n-1}$ are {\bfseries correct} if the following properties hold 
\begin{itemize}
\item[property 1.] For any complete \scan event $S$ and $i< n$, $S$ returns $val(\alpha_i(S))$ at the $i$-th entry (i.e. $val_i(S)=val(\alpha_i(S))$). 

\item[property 2.] For any complete \scan event $S$ and $i< n$, $\neg(S<\alpha_i(S))$.

\item[property 3.] For any \scan event $S$ and any $i< n$, there is no $p_i$-\update event $U$ so that ${\alpha_i(S)<U<S}$.

\item[property 4.] For any two complete \scan events, $S_1$ and $S_2$, and for any $i< n$, if $S_1<S_2$, then ${\alpha_i(S_1)\leq\alpha_i(S_2)}$.

\item[property 5.] For any complete \scan event $S$ and for any $i,j< n$, there is no $p_i$-\update event $U$ so that ${\alpha_i(S)<U<\alpha_j(S)}$.

\item[property 6.] For two complete \scan events, $S_1$ and $S_2$, we define: $S_1<_\alpha S_2$ if $\exists i<n(\alpha_i(S_1)<\alpha_i(S_2))$. We require that for any two complete \scan events $S_1$ and $S_2$, $\neg\Big((S_1<_\alpha S_2)\andd(S_2<_\alpha S_1)  \Big)$
\end{itemize}

\end{definition}
We show that the properties definition \ref{definition-of-correct-functions} provide a necessary and sufficient condition for the correctness of an execution $\tau$.

\begin{theorem}
\label{main-prop-of-section-1-correctness-condition}

$\tau$ is correct iff there are $n$ correct functions $\alpha_0,\dots,\alpha_{n-1}$, where
$$\alphai i. $$
 
\end{theorem}

Before we provide a formal proof for this proposition, we explain the idea behind theorem \ref{main-prop-of-section-1-correctness-condition} and the properties of definition \ref{definition-of-correct-functions}. If $\tau$ is an execution of a snapshot algorithm $\snap$, we want to check if $<$ (defines on the high level events) can be extended into a total order $\prec$ that satisfies the sequential specification. The idea is to relate for each \scan event $S$ and a process $p_i$, some $p_i$-\update event $U_i$ that will be the maximal $p_i$-\update\ that precedes $S$ in $\prec$. This idea defines a function 
$$\alpha_i: \scan\text{ events }\into p_i\text{-\update\ events}$$
by setting $\alpha_i(S)=U_i$. We shall prove that if these functions $\alpha_0,\dots,\alpha_{n-1}$ satisfy the properties of definition \ref{definition-of-correct-functions} (namely, they are correct), then we can extend $<$ into a linear ordering $\prec$ so that:
\begin{enumerate}
\item  for each \scan\ event $S$ and $i<n$, $\alpha_i(S)$ is the maximal $p_i$-\update\ event that precedes $S$ in $\prec$.
\item $\prec$ satisfies the sequential specification (note that this easily stems from the previous claim and from property 1 in definition \ref{definition-of-correct-functions}).  
\end{enumerate}

Now we turn to prove theorem \ref{main-prop-of-section-1-correctness-condition}. The easy direction of our proposition is the ``only if" direction, namely that if $\tau$ is a correct execution, then there are $n$ correct functions $\emptyvectoralpha$, 
$$\alphai i.$$ 
Roughly speaking, for proving this direction we show that the negation of each property in definition \ref{definition-of-correct-functions} indicates a ``bug" in the execution that prevents Linearizability. 

We fix a correct execution $\tau$, a set
$$complete(\tau)\subseteq \class E\subseteq events(\tau)$$ 
and we assume that $(\class E,\prec)$ is a linearization of $\tau$ (i.e. $\prec$ extends $<$ on $\class E$ and satisfies the sequential specification). For any complete \scan event $S$ and $i<n$, we define $\alpha_i(S)$ to be the maximal $p_i$-\update event that precedes $S$ in $\prec$. We claim that these functions satisfy the properties of definition \ref{definition-of-correct-functions}. As an example, we shall prove that property 6 hold, and we leave the straightforward proof of the other properties to the reader.

\begin{proof}
Let $S_1$ and $S_2$ be two complete \scan events. For proving that property 6 hold, assume for a contradiction that $S_1<_\alpha S_2$ and $S_2<_\alpha S_1$. Thus, for some $i,j<n$, $\alpha_i(S_1)<\alpha_i(S_2)$ and $\alpha_j(S_2)<\alpha_j(S_1)$. Assume w.l.o.g. that $S_1\prec S_2$. As $\prec$ extends $<$, we get 
$$\alpha_j(S_2)\prec \alpha_j(S_1).$$
Furthermore, by definition of $\alpha_j$ we have
$$\alpha_j(S_1)\prec S_1.$$

By combining these two observations we conclude
$$\alpha_j(S_2)\prec\alpha_j(S_1)\prec S_1\prec S_2$$
in contradiction to the definition of $\alpha_j$, namely that $\alpha_j(S_2)$ is the maximal $p_j$-\update event that precedes $S_2$ in $\prec$.
\end{proof}

For proving the second direction of theorem \ref{main-prop-of-section-1-correctness-condition} we fix an execution $\tau$ and we argue that if $\alpha_0,\dots,\alpha_{n-1}$ are correct functions, then $\tau$ is correct. We define 
$$\class E=complete (\tau)\cup \{\update \text{ events}\}.$$ 

Clearly, $complete(\tau)\subseteq \class E\subseteq events(\tau)$ and our strategy is to use the functions $\emptyvectoralpha$ to construct a linear ordering $\prec$ on the set of events $\class E$, that extends $<$. Our proof relies on the idea mentioned earlier, namely that $\prec$ is correct if for any complete \scan event $S$ and $i< n$, the maximal $p_i$-\update event that precedes $U$ in $\prec$ is $\alpha_i(S)$. Hence, for a \scan event $S$ and a $p_i$-\update event $U$, we should linearize $U$ before $S$ if $U\leq \alpha_i(S)$, and we should set $S\prec U$ otherwise. However, it is not clear why this approach yields a linear or even a partial ordering. In order to overcome this problem we prove that by extending $<$ in the way described above, we get an acyclic relation (and hence, this relation can be extended to a linear ordering). For the rest of the proof, we speak only about events in $\class E$, the reader may observe that all the \scan events in $\class E$ are complete thus all the \scan events we deal with from this point, are complete.

\begin{definition}
For a \scan event $S$ and a $p_i$-\update event $U$, we define
$U\lhd S \text{ if $U\leq \alpha_i(S)$}$ and $S\lhd U\text{ otherwise.}$  
\end{definition}

As we have said, we shall prove that there are no cycles in $<\cup \, \lhd$. First, we prove this fact only for $\lhd$.

\begin{lemma}\label{less-then-4}
If $X_1\lhd X_2\lhd X_3\lhd X_4$, then $X_1\lhd X_4$.
\end{lemma}

\begin{proof}
There are two possible cases
\begin{itemize}
\item[Case 1.] $X_1$ is a \scan event. We note that if $X\lhd Y$ then one of the events $X,Y$ is a \scan event and the other is an \update event.  Thus, if $X_1$ is a \scan event, then our sequence is of the form
$$S_1\lhd U_1\lhd S_2\lhd U_2$$
where $S_1$ and $S_2$ are \scan events, $U_1$ is a $p_i$-\update event for some $i< n$ and $U_2$ is a $p_j$-{update} event for some $j< n$.
Since $S_1\lhd U_1$, by definition of $\lhd$, $\alpha_i(S_1)<U_1$. Since $U_1\lhd S_2$, we get $U_1\leq \alpha_i(S_2)$. Therefore, 
$$\alpha_i(S_1)<U_1\leq\alpha_i(S_2).$$
As a result $\alpha_i(S_1)<\alpha_i(S_2)$ thus
$$S_1<_\alpha S_2.$$

Now, $S_2\lhd U_2$ indicates that $\alpha_j(S_2)< U_2$. Since $S_1<_\alpha S_2$, by property 6, $\neg(S_2<_\alpha S_1)$ and hence $\alpha_j(S_1)\leq \alpha_j(S_2)$. Therefore, $\alpha_j(S_1)< U_2$ either, and hence $S_1\lhd U_2$ as required.

\item[Case 2.] $X_1$ is $p_i$-\update event for some $i< n$. Thus, our sequence is of the form
$$U_1\lhd S_1\lhd U_2\lhd S_2$$
where $U_1$ is a $p_i$-\update event, $S_1$ and $S_2$ are \scan events and $U_2$ is a $p_j$-\update event for some $j< n$.

$S_1\lhd U_2\lhd S_2$ proves that $\alpha_j(S_1)<U_2\leq \alpha_j(S_1)$. Hence $\alpha_j(S_1)<\alpha_j(S_2)$ and $S_1<_\alpha S_2$ holds. By property 6, $\neg(S_2<_\alpha S_1)$ thus $\alpha_i(S_1)\leq \alpha_i(S_2)$. $U_1\lhd S_1$ indicates that $U_1\leq \alpha_i(S_1)$. Therefore, from $\alpha_i(S_1)\leq\alpha_i(S_2)$ we get that $U_1\leq \alpha_i(S_2)$ as well, and hence $U_1\lhd S_2$ holds as required.

\end{itemize}
\end{proof}

A cycle of length $m>1$ in a binary relation $R$ is a sequence of elements $(X_1,\dots,X_m)$ so that $(X_i,X_{i+1})\in R$ for each $0\leq i<m$ and $X_1=X_m$.

\begin{lemma}\label{lhd-no-cycles}
There are no cycles in $\lhd$.
\end{lemma}

\begin{proof}
Assume for a contradiction that there are cycles in $\lhd$ and consider a cycle of minimal length $(X_1,\dots,X_m)$ where $m>1$. Since $m$ is minimal, by lemma \ref{less-then-4} we conclude that $m<5$. If $m$ is an even integer, then  $X_1$ is a \scan event and $X_m$ is an {update} event, or $X_1$ is an \update event and $X_m$ is a {scan} event. Thus, if $m$ is even then $X_1\neq X_m$. The corollary is that $2\leq m\leq 4$ and $m$ is odd thus $m=3$. Therefore, we get
$$X_1\lhd X_2\lhd X_3\text{ and } X_1=X_3.$$

Now, $X_1$ can be a \scan event or an \update event. First, assume that $X_1$ is a \scan event. Thus, our cycle is of the form
$$S\lhd U\lhd S$$
where $S$ is a \scan event and $U$ is a $p_i$-\update event for some $i< n$. $S\lhd U$ implies that $\alpha_i(S)<U$ while $U\lhd S$ indicates the opposite. Thus, a contradiction has been reached.

It is left to consider the case that $X_1$ is a $p_i$-\update event for some $i< n$. Thus, the sequence is of the form
$$U\lhd S\lhd U$$
where $U$ is a $p_i$-\update event and $S$ is a \scan event. From $U\lhd S$ we conclude that $U\leq \alpha_i(S)$, and from $S\lhd U$ we conclude the opposite. Thus, as in the previous case, This case leads to a contradiction as well.
\end{proof}

So far, we have proved that there are no cycles in $\lhd$. For proving the same for $<\cup\,\lhd$ we need few more lemmas.

\begin{lemma}\label{2-elements}
If $X\lhd Y$ then $\neg(Y<X)$
\end{lemma}  

\begin{proof}
There are two possible cases
\begin{itemize}
\item[Case 1.] $Y$ is a \scan event. Thus, $X$ is a $p_i$-\update event for some $i< n$ and $X\leq \alpha_i(Y)$. By property 2, $\neg(Y<\alpha_i(Y))$ and hence $Y<X$ is impossible.

\item[Case 2.] $Y$ is a $p_i$-\update event for some $i< n$. Thus, $X$ is a \scan event and $\alpha_i(X)<Y$. If we assume that $Y<X$ we get $\alpha_i(X)<Y<X$ in contradiction to property 3, and hence $\neg(Y<X)$. 

\end{itemize}
\end{proof}

\begin{lemma}\label{3-elements}
If $X\lhd Y\lhd Z$, then $\neg(Z<X)$
\end{lemma}  

\begin{proof}
As before, $X$ is either a \scan event or an \update event.
\begin{itemize}
\item[Case 1.] $X$ is a \scan event. Thus, $Y$ is a $p_i$-\update event for some $i< n$ and $Z$ is a \scan event. By definition of $\lhd$, $\alpha_i(X)<Y\leq\alpha_i(Z)$ and hence 
$$\alpha_i(X)<\alpha_i(Z).$$
The assumption $Z<X$ contradicts property 4 thus $\neg(Z<X)$. 

\item[Case 2.] $X$ is a $p_i$-\update event for some $i< n$. In this case $Y$ is a \scan event and $Z$ is a $p_j$-\update event for some $j< n$. By definition of $\lhd$, $X\leq \alpha_i(Y)$ and $\alpha_j(Y)<Z$. Assume for a contradiction that $Z<X$. So, we get $\alpha_j(Y)<Z<X\leq\alpha_i(Y)$ and in particular
$$\alpha_j(Y)<Z<\alpha_i(Y).$$
Since $Z$ is a $p_j$-\update event, our conclusion contradicts property 5, and hence $\neg(Z<X)$ as required.

\end{itemize}
\end{proof}

\begin{lemma}\label{n-elements}
If $X_1\lhd X_2\lhd \dots\lhd X_m$, then $\neg(X_m<X_1)$.
\end{lemma}  

\begin{proof}
Consider a sequence of the form $X_1\lhd X_2\lhd \dots\lhd X_m$. If $m=1$, then the lemma clearly holds since $\neg(X_1<X_1)$. In addition, if $m\geq 2$ by several invocation of lemma \ref{less-then-4} (possibly none) we can construct a sequence $Y_1\lhd\dots\lhd Y_k$ so that 
\begin{itemize}
\item $X_1=Y_1$.
\item $X_m=Y_k$
\item $0<k<4$.
\end{itemize}
if $k=1$ we are done by the previous argument, and if $k\in\{2,3\}$, by lemmas \ref{2-elements} and \ref{3-elements} we conclude that $\neg(Y_k<Y_1)$ and the lemma follows.
\end{proof}

Now we can prove that $<\cup\,\lhd$ can be extended into linear ordering.

\begin{lemma}
There are no cycles in $<\cup\,\lhd$. 
\end{lemma}
\begin{proof}
Assume for a contradiction that there are cycles in $<\cup\,\lhd$ and consider a cycle $(X_1,\dots,X_m)$ of minimal length. Since $<$ and $\lhd$ are both irreflexive, $m>2$. According to lemma \ref{lhd-no-cycles}, for some $i<m$ $X_i<X_{i+1}$ so we may assume w.l.o.g. that $X_1<X_2$. Since $m$ is assumed to be minimal, by the transitivity of $<$, necessarily $X_2\lhd X_3$. We consider two possible cases:
\begin{itemize}

\item[Case 1.] $X_1<X_2\lhd\dots\lhd X_m=X_1$. By lemma \ref{n-elements} $\neg(X_m<X_2)$ and hence for some $e\in X_2,e'\in X_m$, $e<e'$. So, since $X_1<X_2$ and $e\in X_2$ we get $X_1<e<e'\in X_1$. We have concluded that for some $e'\in X_1$, $X_1<e'$ and this is of course, a contradiction. 

\item[Case 2.] $X_1<X_2\lhd \dots\lhd X_k<X_{k+1}$ where $k+1\leq m$. By lemma \ref{n-elements}, for some $e\in X_2,e'\in X_k$, $e<e'$. Thus, $X_1<e<e'<X_{k+1}$ and we get that $X_1<X_{k+1}$, in contradiction to the minimality of $m$.

\end{itemize}
\end{proof}

As there are no cycles in $<\cup\,\lhd$, we conclude that $<\cup\,\lhd$ can be extended into a total ordering $\prec$. Indeed, we define $<^*$ to be the transitive closure of $<\cup\,\lhd$. Since $<\cup\,\lhd$ is an acyclic relation, $<^*$ is a partial ordering and hence can be extended into a total ordering $\prec$.

For completing the proof of theorem \ref{main-prop-of-section-1-correctness-condition} we argue that if $\prec$ is a linear extension of $<\cup\lhd$, then $\prec$ satisfies the sequential specification of the snapshot object (Figure \ref{snapshot-specification}). Since $<\cup\,\lhd$ can be extended into a linear ordering, theorem \ref{main-prop-of-section-1-correctness-condition} stems from the next lemma.

\begin{lemma}
If $\prec$ is a linear extension of $<\cup\,\lhd$, then $\prec$ satisfy the sequential specification.
\end{lemma}

\begin{proof}
Let $S$ be a \scan event, we need to prove that $S$ returns $(val(U_0),\dots,val(U_{n-1}))$ where $U_i$ is the maximal $p_i$-\update event that precedes $S$ in $\prec$. By property 1, $S$ returns $(val(\alpha_0(S)),\dots,val(\alpha_{n-1}(S)))$ thus it suffice to prove that for each $i< n$, $\alpha_i(S)=U_i$.

By definition of $\lhd$, $\alpha_i(S)\lhd S$ and since $\prec$ extends $\lhd$, $\alpha_i(S)\prec S$. Thus,
$$\alpha_i(S)\preceq U_i.$$
If $U$ is a $p_i$-\update event so that $\alpha_i(S)<U$, then $S\lhd U$ and therefore, $S\prec U$. So, since $U_i\prec S$, $\alpha_i(S)<U_i$ is impossible.
However, $U_i$ and $\alpha_i(S)$ are both $p_i$-events, and hence  comparable in $<$ thus $U_i\leq\alpha_i(S)$. As $\prec$ extends $<$, we have $U_i\preceq \alpha_i(S)$. As a result,
$$\alpha_i(S)=U_i$$
follows as required.
\end{proof}

We proved that an execution $\tau$ of a snapshot algorithm $\snap$ is correct iff there are correct functions $\emptyvectoralpha$ so that $$\alphai i.$$ Of course, a snapshot algorithm is correct iff for every execution we can find correct functions as defined in definition \ref{definition-of-correct-functions}. In the next section we prove that when we deal with sb-algorithm it suffice to consider only a small class of executions to ensure the correctness of the algorithm.

\section{Simple Executions}

In this section we prove that if $\snap$ is an sb algorithm, then it is suffice to consider only some of the executions of $\snap$ in order to prove/disprove linearizability. 

The notion of an sb-algorithm is defined in section 2. For an sb-algorithm $\snap$ we define a set of executions, named \bi{simple executions}. In these executions, the \update procedures are invoked with only two different values thus w.l.o.g. we use $0$ and $1$ to denote these values.

\begin{definition}
Let $\tau$ be an execution. We say that $\tau$ is {\bfseries $\bos{(i,j)}$-simple} for two different integers $i,j< n$, if there are $r_i,r_j\in\f{N}$ such that the following hold.
\begin{enumerate}
\item Let $U$ be the $r$-th $p_i$-\update event. If $r<r_i$, then $U$ is an {\sf update($0$)} operation and if $r\geq r_i$, then $U$ is an {\sf update($1$)} operation.

\item In the same way, let $U$ be the $r$-th $p_j$-\update event. If $r<r_j$, then $U$ is an {\sf update($0$)} operation and if $r\geq r_j$, then $U$ is an {\sf update($1$)} operation.

\item All other \update procedure executions are invoked with the value 0. i.e. if $k\neq i,j$ and $U$ is a $p_k$-\update event, then $U$ is an {\sf update($0$)} event.
\end{enumerate}
An execution $\tau$ is {\bfseries simple} if it is $(i,j)$-simple for some different integers $i,j< n$.
\end{definition}

Thus, in simple executions all processes, excluding two of the processes, invoke only {\updatex 0} and \scan procedures. The remaining two processes at first execute {\updatex 0} and \scan operations, and at some point each process stops executing {\updatex 0} operations and starts executing {\updatex 1} operations.  
We claim that in order to prove the correctness of $\snap$, it suffice to prove that any simple execution is correct. This can be deduced from the following proposition:

\begin{proposition}
\label{simple-executions-are-suffice}
Let $\snap$ be an sb algorithm. If there is an incorrect execution $\tau$ of $\snap$, then there is a simple incorrect execution $\tau'$ of $\snap$.
\end{proposition} 

\begin{proof}
We fix an incorrect execution $\tau$ and we shall prove that there is an incorrect simple execution $\tau'$, which is similar (according to definition \ref{similar-executions}) to $\tau$. Recall that $\tau$ admits $n$ functions $\vectoralpha {\tau}$, $$\alpha_i^{\tau}:\text{ complete \scan events} \into p_i \text{ \update events}$$ that satisfy the properties of definition \ref{algorithm-assumption}. In particular, for each complete \scan event $S$,
$$val_\tau(S)=(val_\tau(\alpha_0^\tau(S)),\dots, val_\tau(\alpha_{n-1}^\tau(S))).$$
As $\tau$ is incorrect, the functions $\vectoralpha{\tau}$ are incorrect and hence one of properties 2-6 of definition \ref{definition-of-correct-functions} is violated (note that property 1 holds by the definition of the functions $\vectoralpha{\tau}$). The construction of the simple execution $\tau'$ is according to the property that fails to hold. We consider two cases: the case that property 2 fails and the case that property 6 fails. The cases in which one of properties 3-5 fails to hold are dealt similarly, and the construction of $\tau'$ in these cases is left to the reader. Before we continue we remind that if $\tau$ and $\tau'$ are similar and $E=(s,t)\in \f N\times (\f N\cup\{\infty\})$ a high level event in $\tau$, then it is also a high level event in $\tau'$. Furthermore, if $E$ is a \scan ($\update$) operation in $\tau$, then it is also a \scan ($\update$) event in $\tau'$.

\begin{itemize}
\item[Case 1.] Property 2 does not hold. Thus, for some $i,k\in n$ and a complete $p_k$-\scan event $S$, $S<\alpha_i^\tau(S)$, where $\alpha_i^\tau(S)$ is the $l$-th $p_i$-\update event. Write $U=\alpha_i^\tau(S)$, choose a process i.d. $j\neq i$ and consider the $(i,j)$-simple execution $\tau'$ defined by:
\begin{itemize}
\item $r_i=l$, $r_j=0$.
\item $\tau'$ is similar to $\tau$.
\end{itemize}
As $\tau'$ and $\tau$ are similar, $S$ and $U$ are \scan and \update events in $\tau'$ and by definition, $val_{\tau':i}(S)=val_{\tau'}(U)$. Moreover, as $\tau'$ is $(i,j)$-simple with $r_i=l$, 
$$val_{\tau':i}(S)=val_{\tau'}(U)=1.$$

It is left to prove that $\tau'$ is incorrect. Take $n$-functions $\emptyvectoralpha$
$$\alphai i$$and we shall prove that $\emptyvectoralpha$ are incorrect. Since $\emptyvectoralpha$ are arbitrary, the conclusion is that $\tau'$ is incorrect. 

If $U\leq\alpha_i(S)$, then $S<U\leq\alpha_i(S)$ and property 2 is violated. However, if $\alpha_i(S)<U$, then $val_{\tau'}(\alpha_i(S))=0$ and then property 1 fails to hold as $val_{\tau':i}(S)=1$. Thus, in any case the functions are incorrect and hence $\tau'$ is an incorrect simple execution as required.

\item[Case 2.] Property 6 does not hold. Therefore, there are some complete \scan events $S_1,S_2$ and $i,j< n$ so that 
$$\alpha_i^\tau(S_1)<\alpha_i^\tau(S_2), \ \alpha_j^\tau(S_2)<\alpha_j^\tau(S_1).$$
Assume that $\alpha_i^\tau(S_1)$ is the $t_1$-th $p_i$-\update event and that $\alpha_j^\tau(S_2)$ is the $t_2$-th $p_j$-\update event. Write $I_1=\alpha_i^\tau(S_1)$, $I_2=\alpha_i^\tau(S_2)$, $J_1= \alpha_j^\tau(S_1)$ and $J_2=\alpha_j^\tau(S_2).$

Let $\tau'$ be an $(i,j)$-simple execution so that
\begin{itemize}
\item $\tau'$ is similar to $\tau$.
\item $r_i=t_1+1$, $r_j=t_2+1$.
\end{itemize}

As $\tau'$ and $\tau$ are similar, note that $I_1,I_2$ are $p_i$-\update events in $\tau'$, $J_1,J_2$ are $p_j$-\update events in $\tau'$, and $S_1$, $S_2$ are complete \scan events in $\tau'$. Furthermore, since $\tau'$ and $\tau$ are similar and since $\tau'$ is $(i,j)$-simple with $r_i=t_1+1$, $r_j=t_2+1$ the following hold in $\tau'$:
\begin{enumerate}
\item $I_1<I_2$, $J_2<J_1$.
\item $val_{\tau':i}(S_1)=val_\tau'(I_1)=0$, $val_{\tau':j}(S_1)=val_\tau'(J_1)=1$.
\item $val_{\tau':i}(S_2)=val_\tau'(I_2)=1$, $val_{\tau':j}(S_2)=val_\tau'(J_2)=0$.

\end{enumerate}

As in the previous case, let $\emptyvectoralpha$ be $n$ functions 
$$\alphai i$$
and we need to verify that these functions are incorrect. If property 1 fails to hold we are done, and otherwise we have 
\begin{enumerate}
\item $val_{\tau'}(\alpha_i(S_1))=0$, $val_{\tau'}(\alpha_j(S_1))=1$.

\item $val_{\tau'}(\alpha_i(S_2))=1$, $val_{\tau'}(\alpha_j(S_2))=0$.

\end{enumerate} 

Since $I_1$ is the last $p_i$-\update event invoked with the value $0$ we conclude

\begin{equation}
\label{first}
\alpha_i(S_1)\leq I_1< \alpha_i(S_2).
\end{equation}

Similarly, since $J_2$ is the last $p_j$-\update event invoked with the value $0$ we conclude

\begin{equation}
\label{second}
\alpha_j(S_2)\leq J_2< \alpha_j(S_1).
\end{equation}

Equations \ref{first} and \ref{second} imply that $$\alpha_i(S_1)<\alpha_i(S_2) \andd \alpha_j(S_2)<\alpha_j(S_1)$$ thus in this case, property 6 fails to hold.
\end{itemize}
\end{proof}

By proposition \ref{simple-executions-are-suffice} we conclude.
\begin{theorem}
\label{reduction-to-simple-executions}
A snapshot algorithm $\snap$ is correct iff all the simple executions of $\snap$ are correct.
\end{theorem}

Now, let $\snap$ be finite sb-snapshot algorithm. By our corollary, in order to check the correctness of $\snap$ it suffice to check only executions in which the \update events are invoked with the values $0$ or $1$. Recall that under this restriction only finite states of $\snap$ are reachable as $\snap$ is assumed to be finite. In \cite{AMP} Alur et al. proved that linearizability is decidable for algorithms with finite states and a fixed number of processes. Thus, based on Alur et al. result we have
\begin{corollary}
It is decidable to determine if a finite sb-snapshot algorithm is correct.
\end{corollary}

\ignore{
However, the time complexity of the verification method in \cite{} is exponential in the number of (reachable) states of $\snap$. In the next section we present a novel approach. The time complexity of our method is polynomial in the number of states of $\snap$.}

\section{Conclusions}

In section 3 we defined a necessary and sufficient condition for correctness of executions of snapshot algorithms. Our condition relays upon the existence of functions $\alpha_0,\dots,\alpha_{n-1}$ between \scan event and \update events by $p_0,\dots,p_{n-1}$ respectively. A programmer is likely to be able to define these functions for hers implementation. Our condition provides programmers with a framework for implementing snapshot algorithms and proving correctness of snapshot algorithms.

We have defined the schedule-based notion. Here, we focus only on snapshot implementations but the schedule-based notion can be applied on other objects as well, such as stack, queue, etc. We use this notion to look at simple executions. Since verifying linearizability is EXSPSPACE-complete, it is important to seek for natural assumptions that concurrent implementations satisfy. This kind or research can potentially lead to constructing feasible verification techniques.

We proved that if $\snap$ is an sb-snapshot algorithm, then $\snap$ is correct iff all its simple executions are correct. Relaying on our theorem, we concluded that known verification techniques (for example \cite{AMP}) can be applied on finite sb-snapshot implementations. Recall that when we say that an implementation is finite, we mean that it is finite only when simple execution are assumed. Thus, our result is crucial for applying verification techniques such as the one in \cite{AMP}.

We consider this paper as a starting point for varied further research, as it raises many questions. Since verifying linearizability is decidable but not feasible, we suggest two approaches for overcoming this problem. First, we suggest to adapt assumptions on the algorithm verified to be correct. For example, in \cite{Vaf} the verification techniques use an assumption on the linearization points. However, this assumption exclude some known algorithms, for example the queue implementation in \cite{Lin}. We defined the schedule-based property and we argue that this property is very natural to assume. Second, we suggest to look at unique objects. Here we focus on the well-known snapshot object and we hope that our results can lead to designing polynomial verification tools for snapshot implementations. 

We set three main directions for further research in view of our ideas and results
\begin{itemize}
\item Finding conditions that resemble the properties in definition \ref{definition-of-correct-functions} for other objects and data structures.

\item Applying the notion of sb-algorithms on other objects, and finding a corresponding variants of our simple executions.

\item Use our reduction to simple executions to design polynomial automatic verification of sb-snapshot implementations. Replicating this approach for other objects and data structures.  
\end{itemize}

\end{document}